%% file: output.tex
\documentclass[manuscript,nonacm]{acmart}
\usepackage[version=4]{mhchem} 
\usepackage{caption}

\usepackage{lipsum} 
\usepackage{colortbl}
\usepackage{xcolor}  

\usepackage{amssymb} 

\usepackage{amsmath}
\usepackage{algorithm}
\usepackage{amsthm}
\usepackage{mathtools} 
\usepackage{algpseudocode} 
\usepackage{enumitem}
\usepackage{booktabs}
\usepackage{subcaption}
\usepackage{textcomp}

\usepackage{multirow}
\usepackage{pifont}
\newcommand{\cmark}{\ding{51}}  
\newcommand{\xmark}{\ding{55}}  
\usepackage{threeparttable}

\newtheorem{theorem}{Theorem}[section]

\newtheorem{lemma}[theorem]{Lemma}
\theoremstyle{definition}
\newtheorem{definition}{Definition}

\newcommand{\name}{\textit{PrivyWave}}

\DeclareSymbolFont{extraup}{U}{zavm}{m}{n}
\DeclareMathSymbol{\varheart}{\mathalpha}{extraup}{86}
\DeclareMathSymbol{\vardiamond}{\mathalpha}{extraup}{87}

\begin{document}

\title{\name{}: Privacy-Aware Wireless Sensing of Heart Beat}

\author{Yixuan Gao}
\authornote{Both authors contributed equally to this research.}
\email{yixuan@cs.cornell.edu}
\orcid{TODO}
\affiliation{%
  \institution{Cornell Tech}
  \city{New York}
  \state{New York}
  \country{USA}
}

\author{Tanvir Ahmed}
\authornotemark[1]
\email{tanvir@infosci.cornell.edu}
\orcid{0000-0002-9468-5033}
\affiliation{%
  \institution{Cornell Tech}
  \city{New York}
  \state{New York}
  \country{USA}
}

\author{Zekun Chang}
\email{zekunchang@infosci.cornell.edu}
\orcid{TODO}
\affiliation{%
  \institution{Cornell Tech}
  \city{New York}
  \state{New York}
  \country{USA}
}

\author{Thijs Roumen}
\email{thijs.roumen@cornell.edu}
\orcid{TODO}
\affiliation{%
  \institution{Cornell Tech}
  \city{New York}
  \state{New York}
  \country{USA}
}

\author{Rajalakshmi Nandakumar}
\email{rajalakshmi.nandakumar@cornell.edu}
\orcid{TODO}
\affiliation{%
  \institution{Cornell Tech}
  \city{New York}
  \state{New York}
  \country{USA}
}

\renewcommand{\shortauthors}{Gao, et al.}

\begin{abstract}
Wireless sensing technologies can now detect heartbeats using radio frequency and acoustic signals, raising significant privacy concerns. Existing privacy solutions either protect from all sensing systems indiscriminately preventing any utility or operate post-data collection, failing to enable selective access where authorized devices can monitor while unauthorized ones cannot. We present a key-based physical obfuscation system, PrivyWave, that addresses this challenge by generating controlled decoy heartbeat signals at cryptographically-determined frequencies. Unauthorized sensors receive a mixture of real and decoy signals that are indistinguishable without the secret key, while authorized sensors use the key to filter out decoys and recover accurate measurements. Our evaluation with 13 participants demonstrates effective protection across both sensing modalities: for mmWave radar, unauthorized sensors show 21.3 BPM mean absolute error while authorized sensors maintain a much smaller 5.8 BPM; for acoustic sensing, unauthorized error increases to 42.0 BPM while authorized sensors achieve 9.7 BPM. The system operates across multiple sensing modalities without per-modality customization and provides cryptographic obfuscation guarantees. Performance benchmarks show robust protection across different distances (30-150 cm), orientations (120° field of view), and diverse indoor environments, establishing physical-layer obfuscation as a viable approach for selective privacy in pervasive health monitoring.

\end{abstract}




\begin{teaserfigure}
  \includegraphics[width=\textwidth]{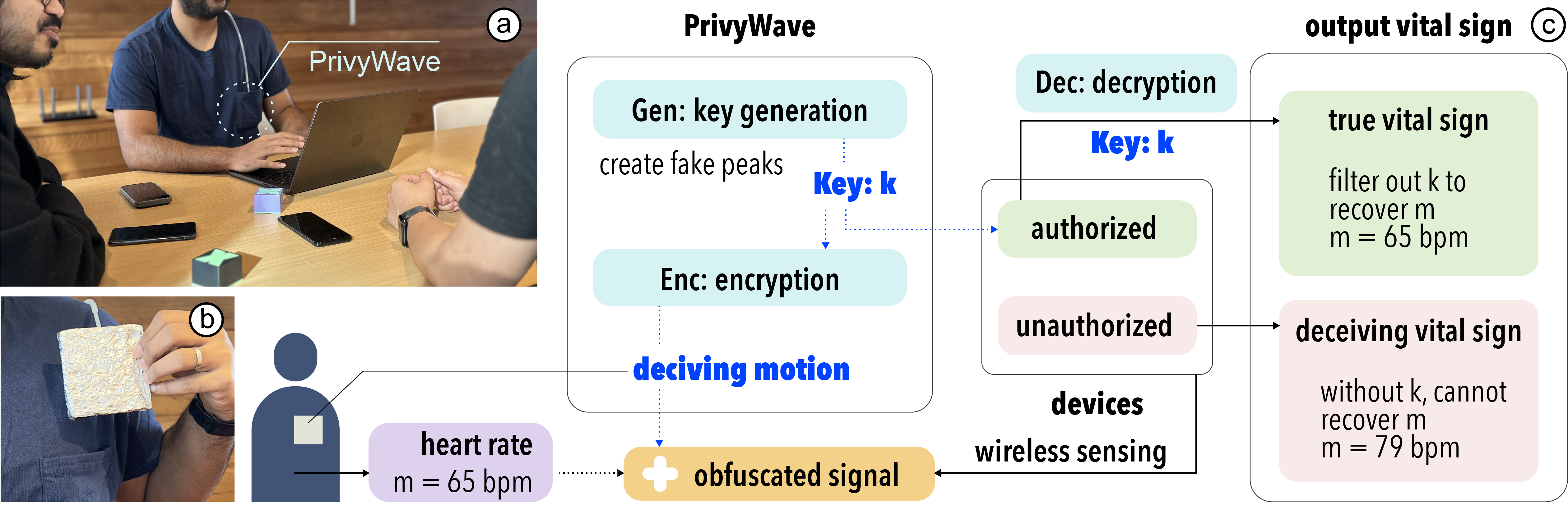}
  \caption{\name{} system overview. (a) A person wearing \name{} in office, surrounded by ubiquitous devices (e.g., mobile phones, laptops, speakers, Wi-Fi router) that could potentially function as wireless senors to pick up vital signal from people without them noticing. (b) \name{} pneumatic actuator device. (c) System workflow: key generation creates decoy signal frequencies $k$, encryption superimposes decoy motion on true vital sign $m$ using the actuator, and authorized devices use $k$ to decrypt and recover the true signal while unauthorized devices observe one of the decoy signals.}
  \Description{xxxxxx}
  \label{fig:teaser}
\end{teaserfigure}

\maketitle

\input{Sections/Introduction}
\input{Sections/RelatedWork_new}

\input{Sections/Background_new}

\input{Sections/SystemDesign}
\input{Sections/Implementation}

\input{Sections/Experiments}

\input{Sections/Discussion}

\bibliographystyle{ACM-Reference-Format}
\bibliography{references}


\end{document}

%% file: Sections/Introduction.tex
\section{Introduction}

Wireless sensing technologies have advanced rapidly in recent years. Systems using acoustic~\cite{nandakumar2015contactless}, WiFi~\cite{abdelnasser2015ubibreathe}, and millimeter-wave~(mmWave)~\cite{yang2017vital} sensors can now monitor vital signs without any physical contact, detecting chest movements as small as few millimeters to measure breathing and heartbeat. These capabilities enable critical healthcare applications, such as detecting sleep apnea~\cite{nandakumar2015contactless}, identifying opioid overdoses~\cite{nandakumar2019opioid}, even detecting cardiac arrest and calling emergency services~\cite{chan2019contactless} all without requiring users to wear sensors.

However, these same technological capabilities that promise healthcare benefits also create privacy risks. Any ambient device such as a WiFi router, smart speaker, home assistant or someone's smartphone, can potentially become a covert physiological signal monitor, silently extracting vital signs without user awareness or consent. This threat is particularly concerning because vital signs reveal far more than basic metrics, they could expose emotional states~\cite{zhao2016emotion}, or stress levels~\cite{ha2021wistress} that may even indicate mental health conditions. In addition, some wireless signals are invisible and penetrate walls and clothing, creating an omnipresent surveillance risk which is hard for users to detect or avoid. So as wireless sensing becomes ubiquitous, developing solutions that maintain its practical benefits while preventing unauthorized surveillance is a necessity.



Existing privacy protection approaches for wireless signals fall into two categories: post-collection data processing and real-time protection. Post-collection approaches~\cite{shi2010prisense,hassan2019differential,liu2023application} protect data after acquisition through techniques such as data aggregation \cite{shi2010prisense}, differential privacy~\cite{hassan2019differential}, and signal tampering~\cite{liu2023application}. However, these methods cannot prevent unauthorized sensing in the first place, as adversaries can collect raw signals in real-time before any post-collection protection is applied. For real-time protection, jamming methods~\cite{chen2020wearable, yu2025dynamic, qiao2025nusguard} inject noise to degrade signal quality, but they are modality-specific—protecting against one type of sensor (e.g., acoustic) but not others (e.g., radio-frequency or RF). Anti-Sensing~\cite{oshim2025anti} uses wearable oscillators to mislead radar-based heartbeat detection, but blocks all sensors indiscriminately. VitalHide~\cite{gao2025vitalhide} presented the first approach for selective protection using vibration-based obfuscation. However, this work remained conceptual, lacking theoretical foundations, formal privacy analysis, and systematic validation.

We present \name{}, a combined software and hardware solution that protects users from unauthorized wireless monitoring while preserving the full utility of authorized sensing. Our central idea is to enable private wireless sensing using decodable physical-layer obfuscation technique. We generate controlled, physical actuation co-located with the user's body, which mimics the periodic motion of vital signs (e.g., a heartbeat) to create plausible decoy signals. Since wireless sensing systems operate by detecting such periodic motions, an unauthorized sensor observing the user perceives a composite \textit{obfuscated} signal and is unable to distinguish the user's true vital sign from the decoys. Conversely, an authorized system, possessing a shared cryptographic key, can computationally identify and filter out these decoy signals, thereby recovering the user's true signal with high fidelity. Our system overview is given in Fig. \ref{fig:teaser}.

This method has two main advantages i) The obfuscation is done by generating fake motions and hence this method is agnostic to different wireless modalities and frequencies such as acoustic, RF (WiFi, mmWave). ii) The generated decoy signals are similar to legitimate physiological signals which the unauthorized devices cannot distinguish.

In this work, first,  we design a cryptographic framework for physical-layer obfuscation that provides a privacy guarantee. We define (1) a key generation procedure for creating valid decoy signal frequencies to obfuscate the true signal, (2) an encryption algorithm that physically generates these decoys through actuation and obfuscates true vital signal, and finally (3) a decryption algorithm that enables authorized devices to recover the true signal by filtering out known decoy frequencies with they key. 
We mathematically show that unauthorized attackers gain negligible distinguishing advantage from random guessing among the decoy frequencies.
Second, we implement this framework through a compact silicone-based pneumatic prototype design that generates heartbeat-like actuation patterns while fitting inside the user's pocket. Third, we experimentally validate that our system works across multiple sensing modalities: frequency-modulated continuous wave~(FMCW) mmWave radars and acoustic sonars, demonstrating modality-agnostic protection.

We validate \name{} for mmWave and acoustic systems through a user study with \textcolor{black}{13} participants. The results show that detection error and standard deviation for unauthorized sensor is significantly higher ($p<.001$) than for authorized devices. For heart rate detection, the average mean absolute error (MAE) for mmWave unauthorized and authorized devices are 21.3 and 5.8 BPM, while for acoustic are 42 and 9.7 BPM respectively. We also show the effectiveness of \name{} in different environments as well as different radar range and orientation. 
These results demonstrate that \name{} enables a new paradigm for wireless vital sign sensing, one where users can benefit from continuous wireless health monitoring while not having to worry about unwanted surveillance.



In this work, we explore the possibility of preserving the utility of wireless sensing while giving users agency over their privacy. From developing a cryptographic obfuscation framework for physical signals to building actual hardware systems, this work demonstrates that functionality and strong privacy protection can coexist in wireless sensing systems. To summarize, our key contributions are:
\begin{itemize}
    \item We design and build the first motion-based, modality-agnostic, and privacy-preserving physical layer obfuscation wireless sensing system for vital sign monitoring called \name{} that protects against unauthorized monitoring while keeping the full utility of authorized devices.
    \item We provide a mathematical bound for the unauthorized attacker's distinguishing advantage for the designed physical layer obfuscation system.
    \item We demonstrate the effectiveness of \name{} through a comprehensive user study and micro-benchmarks, which shows a significantly higher detection error for unauthorized devices.
\end{itemize}

%% file: Sections/RelatedWork_new.tex
\section{Related Work}

We discuss existing approaches for wireless vital sign privacy protection. We first review the capabilities of wireless vital sign sensing systems and their privacy implications. We then examine the current SOTA privacy protection solutions for wireless sensing, categorized into post-processing and real-time protection approaches. We then focus particularly on real-time protection systems, which can be further divided into jamming-based and obfuscation-based methods.

\begin{table*}[t]
\centering
\caption{Comparison of Privacy Protection Approaches}
\label{tab:related_work_comparison}
\small
\begin{tabular}{@{}llcccccc@{}}
\toprule
\textbf{Approach} & \textbf{Target} & \textbf{Real-Time} & \textbf{Selective} & \textbf{Multi-} & \textbf{Security} & \textbf{Auth.} \\
 & \textbf{Domain} & \textbf{Protect.} & \textbf{Protect.} & \textbf{Modal} & \textbf{Guarantee} & \textbf{Access} \\
\midrule
\multicolumn{7}{l}{\textit{Post-Collection Processing}} \\
\midrule
PriSense~\cite{shi2010prisense} & Sensor Data & \xmark & \xmark & \cmark & \cmark & N/A \\
Diff. Privacy~\cite{hassan2019differential} & Sensor Data & \xmark & \xmark & \cmark & \cmark & N/A \\
mmFilter~\cite{liu2023application} & mmWave Radar & \xmark & \cmark$^{\ddag}$ & \xmark & \xmark & N/A \\
\midrule
\multicolumn{7}{l}{\textit{Real-Time Jamming}} \\
\midrule
Wearable Jammer~\cite{chen2020wearable} & Audio & \cmark & \xmark & \xmark & \xmark & \xmark \\
Dynamic Jamming~\cite{yu2025dynamic} & Audio & \cmark & \xmark & \xmark & \xmark & \xmark \\
NUSGuard~\cite{qiao2025nusguard} & Voice & \cmark & \cmark$^{\dag}$ & \xmark & \xmark & \cmark$^{\dag}$ \\
\midrule
\multicolumn{7}{l}{\textit{Real-Time Obfuscation}} \\
\midrule
RF-Protect~\cite{shenoy2022rf} & Human Tracking & \cmark & \xmark & \xmark & \xmark & \xmark \\
Radar Obfus.~\cite{argyriou2023obfuscation} & Activity & \cmark & \xmark & \xmark & \xmark & \xmark \\
Anti-Sensing~\cite{oshim2025anti} & Vital Signs & \cmark & \xmark & \cmark & \xmark & \xmark \\
VitalHide~\cite{gao2025vitalhide} & Vital Signs & \cmark & \cmark$^{*}$ & \cmark$^{*}$ & \xmark & \cmark$^{*}$ \\
\midrule
\textbf{\name{} (Ours)} & \textbf{Vital Signs} & \cmark & \cmark & \cmark & \cmark & \cmark \\
\bottomrule
\multicolumn{7}{l}{\footnotesize $^{*}$Proof-of-concept only; $^{\dag}$Temporal selectivity only; $^{\ddag}$Application-level, not device-level} \\
\multicolumn{7}{l}{\footnotesize \cmark: Supported; \xmark: Not supported; N/A: Not applicable} \\
\end{tabular}
\end{table*}

\subsection{Wireless Vital Sign Sensing and Privacy Implications}

Wireless vital sign sensing has evolved into reliable systems with high accuracy. Many smart devices have gained the ability to accurately detect vital signs without physical contact. Nandakumar et al.~\cite{nandakumar2015contactless} demonstrated that smartphones can be turned into vital sign monitors by leveraging the acoustic sensors in the device. WiFi-based approaches extract breathing and heartbeat from the RF signals reflected by the subject~\cite{abdelnasser2015ubibreathe, liu2015tracking}. The mmWave sensors, commonly deployed for presence detection, can also measure breathing and heartbeat with high precision~\cite{yang2017vital, gong2021rf}.These sensing capabilities have enabled beneficial applications in everyday life. Researchers have developed systems for sleep apnea detection~\cite{nandakumar2015contactless}, which can alert users to potentially dangerous breathing interruptions during sleep. Contactless vital sign monitoring has been proposed for overdose detection~\cite{nandakumar2019opioid}, enabling earlier intervention in emergency situations. Cardiac arrest detection systems~\cite{chan2019contactless} can automatically alert emergency services when abnormal vital signs are detected. However, vital signals carry far more information than just physiological measurements. Stress levels can be inferred from vital signals~\cite{ha2021wistress}, emotional states can be detected from combined respiratory and cardiac patterns~\cite{zhao2016emotion}, and cognitive load during mental tasks can be assessed through cardiovascular responses~\cite{solhjoo2019heart}. As sensing accuracy continues to improve, the precision of these inferences also increases, making the privacy risks more severe.

This dual nature of wireless vital sign detection, which simultaneously allows beneficial health monitoring and creates unprecedented privacy risks, motivates the need for privacy protection mechanisms that can distinguish between authorized and unauthorized detection.

\subsection{Privacy Protection for Wireless Sensing}
We categorize existing privacy protection approaches for wireless signals into two main strategies: post-collection data processing and real-time signal protection.

\subsubsection{Post-Collection Data Processing}
There are a set of algorithms that focuses on protecting data after it has been collected by sensors. These approaches assume that sensing has already occurred and apply privacy-preserving techniques during subsequent processing, transmission, or storage stages. PriSense~\cite{shi2010prisense} introduced privacy-preserving data aggregation for sensor networks using data slicing, where each sensor splits its reading into multiple shares distributed to randomly selected cover nodes, protecting individual data unless the aggregation server colludes with all cover nodes. Differential privacy approaches~\cite{hassan2019differential} provide stronger theoretical guarantees by adding calibrated noise to sensor data datasets, ensuring that the presence or absence of any individual's data cannot be reliably determined while preserving aggregate statistics. For mmWave radar-based sensing, mmFilter~\cite{liu2023application} applies signal reversion techniques that tamper with radar data after collection but before transmission to sensing processors. While these post-collection methods provide valuable privacy protections for captured data, they share a fundamental limitation: they cannot prevent unauthorized sensing in the first place. An adversary can still collect the raw signals before any privacy protection is applied.

\subsubsection{Real-Time Signal Protection}

To enable protection against unauthorized sensing in real-time, researchers have developed technologies that mainly fall into two categories: jamming and obfuscation.

\paragraph{Jamming-Based Protection}
Jamming approaches inject noise or interference to degrade signal quality(decrease SNR) for unauthorized receivers. For audio privacy, ultrasonic jamming exploits microphone nonlinearity: high-frequency ultrasonic signals cause microphones to produce audible-range noise, corrupting recordings. Chen et al.\cite{chen2020wearable} developed a wearable bracelet with 24 ultrasonic transducers providing omnidirectional microphone jamming. Yu et al.\cite{yu2025dynamic} improved efficiency through adaptive jamming that analyzes speech characteristics in real-time and generates time-frequency interference patterns matched to audio content. NUSGuard~\cite{qiao2025nusguard} introduces temporal selectivity by detecting when users interact with authorized voice assistants, temporarily disabling jamming during those interactions.
However, jamming approaches face fundamental limitations. Jamming is inherently modality-specific: it protects against one type of sensing signal, for example acoustic, but not RF or other sensing modalities, and vice versa, and we do not have control over what modality an unauthorized sensor uses.

\paragraph{Obfuscation-Based Protection}
Obfuscation techniques inject plausible decoy information rather than noise, making it difficult to distinguish true signals from fake alternatives\cite{brunton2015obfuscation}. The core principle is to hide the real signal among believable decoys.
For human tracking, RF-Protect~\cite{shenoy2022rf} introduced a hardware reflector-based approach that injects phantom humans into device-free tracking systems. The system uses specially designed reflectors to modify radio waves and create reflections at arbitrary locations, combined with a generative model to create realistic human trajectories. For activity recognition, Argyriou~\cite{argyriou2023obfuscation} demonstrated that synthetic motion patterns can obfuscate human micro-Doppler signatures in passive radar. For vital sign protection, Anti-Sensing~\cite{oshim2025anti} uses wearable oscillators that generate motion patterns mimicking natural cardiac motion, creating decoy signals that mislead radar-based heartbeat detection. While the approach successfully generates realistic oscillatory patterns, it blocks all radar sensors indiscriminately without providing selective access for authorized monitoring. VitalHide~\cite{gao2025vitalhide} recently explored whether phone vibrators and smart textile actuators could generate physical obfuscation for vital signs, creating decoy heartbeat signals to confuse unauthorized sensors while potentially allowing authorized devices to filter them out. The proof-of-concept showed that vibration-based obfuscation could reduce unauthorized detection accuracy. However, this work remained at the conceptual demonstration stage, lacking theoretical foundations, formal security analysis, and systematic validation of the selective access mechanism across different sensing modalities.

Our work builds on the obfuscation approach but provides key advances: (1) a formal cryptographic framework for key-based selective access, including key generation, encoding, and decoding algorithms, (2) security analysis with provable guarantees for authorized and unauthorized scenarios, (3) wearable hardware designs that realize this framework on the human body, (4) systematic evaluation across multiple sensing modalities (mmWave and acoustic), and (5) performance benchmarks under varying physical conditions. Our system enables authorized monitoring while protecting against unauthorized sensing through cryptographic key-based decoy filtering.

%% file: Sections/Background_new.tex
\section{background - FMCW radar}

Frequency Modulated Continuous Wave (FMCW)~\cite{stove1992linear} has emerged as the preferred technology for contactless vital sign monitoring due to its ability to detect sub-millimeter movements with high precision without the need for sampling signal at carrier frequency. This capability enables FMCW radar to be implemented across a variety of signal modalities. In this section, we cover the FMCW radar theory and signal processing techniques used for extracting vital signs. We will then explain the design of \name that can protect the subjects from unauthorized sensors that run these algorithms to extract physiological signals.

\subsection{FMCW Theory}
The FMCW system transmits a chirp signal $s_{tx}(t)$, a sinusoid whose frequency increases linearly over time:
\begin{equation}
s_{tx}(t) = A_t \exp\left(j2\pi\left(f_c t + \frac{B}{2T_{chirp}}t^2\right)\right)
\end{equation}
where $f_c$ is the carrier frequency, $B$ is the bandwidth, and $T_{chirp}$ is the chirp duration. When this chirp reflects off a target at distance $d(t) = d_0 + x(t)$, where $x(t)$ represents chest displacement from breathing and heartbeat, the received signal experiences a time delay $\tau = 2d(t)/c$:
\begin{equation}
s_{rx}(t) = A_r \exp\left(j2\pi\left(f_c (t-\tau) + \frac{B}{2T_{chirp}}(t-\tau)^2\right)\right)
\end{equation}

After dechirping (mixing transmitted and received signals), the intermediate frequency (IF) signal contains two critical phase components:
\begin{equation}
s_{IF}(t) = A_{IF} \exp\left(j\underbrace{\frac{4\pi B d(t)}{cT_{chirp}}t}_{\text{beat frequency}} + j\underbrace{\frac{4\pi f_c d(t)}{c}}_{\text{phase from displacement}}\right)
\end{equation}

These two phase terms serve distinct purposes in vital sign sensing. The first term, $\frac{4\pi B d(t)}{cT_{chirp}}t$, represents a beat frequency that is proportional to the target distance $d(t)$. This term enables range separation: by applying FFT over samples within a single chirp (typically 50-100 microseconds), we obtain a range profile where echoes from different distances produce distinct frequency peaks. The range resolution is determined by the bandwidth: $\Delta R = c/(2B)$. 

The second term, $\phi(t) = 4\pi f_c d(t)/c = 4\pi d(t)/\lambda$, is the phase component that encodes fine-grained displacement information. Within a single chirp duration, this phase remains approximately constant since the chest displacement $x(t)$ changes negligibly over microseconds. However, across multiple chirps (frame period typically 20-50 milliseconds), this phase evolves as the chest moves due to breathing and heartbeat.  By tracking this phase evolution across consecutive chirps, we extract the time-varying displacement signal:
\begin{equation}
x(t) = \frac{\lambda}{4\pi} \phi(t)
\end{equation}

\subsection{Signal Processing Methods for Vital Sign Extraction}

Once displacement $x(t)$ is extracted, various methods process it to isolate vital signs.

\paragraph{Non-Learning-Based signal processing Methods}
Peak detection algorithms~\cite{alizadeh2019remote} identify local maxima and minima in $x(t)$ corresponding to breathing cycles or cardiac phases, determining instantaneous rates from time intervals between peaks. FFT-based methods apply bandpass filtering to isolate specific frequency ranges (0.1--0.5~Hz for breathing, 0.8--2.0~Hz for heart rate), identifying peak frequencies as vital sign rates~\cite{alizadeh2019remote,lee2019novel}. Time-frequency decomposition techniques ~EMD/EEMD~\cite{xu2022non}, VMD~\cite{zhao2023accurate}, and wavelet transforms~\cite{wang2023slprof} decompose the displacement signal into simpler signal components for analysis.

\paragraph{Learning-Based Methods}
Deep learning approaches train neural networks on the features that include the displacement $x(t)$~\cite{wang2023here,kim2024heartbeatnet} to estimate vital signs. These methods handle complex scenarios including multi-person environments~\cite{wang2020vimo} and low SNR conditions, but require substantial labeled data and may not generalize across deployment environments.

\subsection{Different modalities of wireless sensing}
While both RF and acoustic systems apply identical FMCW principles and similar signal processing methods described above, they operate at different ranges with distinct trade-offs. Millimeter-wave radar transmits and receives RF signal of frequency range 60--77 GHz that travel at the speed of light. They achieve a millimeter level range resolution and a range of 3 to 5 meters and can penetrate clothes. However they require specialized hardware. WiFi based RF radars operate at 2.4 and 5 Ghz and has similar properties where they can penetrate evn through walls. However they also require specialized hardware such as expensive USRP. Acoustic FMCW (18--22 kHz) achieves millimeter level range resolution on commodity smartphones as the speed of sound is much lower than RF($c_{sound} \approx 343$ m/s), but the range of the system is limited to 0.5--1.5 meters with poor penetration and high noise susceptibility. The choice reflects deployment context: acoustic systems democratize sensing through ubiquitous devices for close-proximity applications, while RF based systems such as WiFi and mmWave enables through-clothing monitoring at larger distances essential for privacy-sensitive scenarios.

In order to build a system that can hide from unauthorized sensors without depending on what modality they use and what algorithms they used, we need to build physical obfuscation: co-located fake heartbeat signals on the human body. 

Despite the differences in the algorithmic diversity, all these systems record the reflections of custom frequency signals and analyze the variations caused by sub-centimeter motion generated from the human body. Hence, a co-located heartbeat generator with the human would obfuscate all sensing algorithms. For example, peak-detection algorithms will be confused by the peaks induced by the fake heartbeat, while frequency-dependent algorithms will be confused by the new frequency components added by the fake heart rate.

%% file: Sections/SystemDesign.tex
\section{System Design}
\subsection{System Overview}
To guarantee privacy against different sensing modality (e.g., acoustic and mmWave) that could run any signal processing algorithm, we designed \name{}, which operates by generating controlled heartbeat-like motions using a pneumatic-based device prototype. We carefully design the obfuscation mechanism so that a single device can generate multiple different-frequency heartbeat-like signals in real time. From an unauthorized sensor's perspective, it will detect multiple heartbeat signals with the real one immersed among them, while authorized sensors that possess the cryptographic key can recover the true signal from the obfuscated composite signal. 

We introduce our system by first formally formulating the problem and defining the threat model. We then present the obfuscation algorithm that describes how we generate the obfuscation signals, how these signals are encoded, and how authorized sensors can decode them. This is followed by a formal privacy bound analysis. Finally, we show how we implement the obfuscation scheme on a hardware prototype.

\subsection{Problem Formulation}

Our goal is to design a system that could enable provable secure wireless vital sign monitoring where authorized devices can accurately monitor within a negligible error, while unauthorized devices are prevented from extracting meaningful information in real-time. 

From the lens of cryptography theory, the task of privacy-preserving wireless vital sign sensing system can be seen as a "secure communication" event. The user ($\mathcal{U}$) whose vital signs is to be protected becomes the \textit{sender} (aka Alice), the vital sign is the \textit{private message} ($m$) that the user is trying to send to an authorized wireless sensing system (aka Bob), who becomes the \textit{receiver} ($\mathcal{V}$). The unauthorized wireless sensing systems, who try to intercept the \textit{private message} becomes the \textit{adversary} ($\mathcal{A}$) (aka Eve). In this work, we take a physical layer obfuscation approach, which \textit{obfuscates} $m$ into a ciphertext $c$, to provide privacy to the user.

\subsection{Threat Model}
We model the adversary $\mathcal{A}$ as a passive eavesdropper seeking to measure the user's vital signs through wireless sensing. We assume $\mathcal{A}$ is computationally bounded (a probabilistic polynomial-time adversary) and operates under the following conditions:

\begin{itemize}
    \item \textbf{System Knowledge:} $\mathcal{A}$ has complete knowledge of the system design, algorithms, and probability distributions (per Kerckhoffs's Principle). Security relies solely on the secrecy of the session-specific cryptographic key ($k$), which defines the decoy signal frequencies.
    
    \item \textbf{Sensing Capabilities:} $\mathcal{A}$ can deploy arbitrary wireless sensing equipment (e.g., mmWave radar, acoustic FMCW sensors) with any number of antennas and apply any signal processing algorithm to extract vital signs.
    
    \item \textbf{Spatial Resolution Limit:} $\mathcal{A}$ cannot spatially separate the user's true vital sign motion from the co-located decoy motion generated by \name{}. Both signals originate from the same location on the user's body and are perceived by the sensor as a single, superimposed signal with multiple frequency components.
    
    \item \textbf{Passive Attack Constraint:} $\mathcal{A}$ is restricted to passive observation only. The model excludes active attacks such as stimulus-response probing (e.g., inducing a physical startle to identify reactive vs. non-reactive signals). Defending against such active attacks is left for future work.
\end{itemize}

Under this threat model, our goal is to prevent adversary $\mathcal{A}$ from distinguishing the true vital sign frequency from obfuscation frequencies with probability better than random guessing.

\subsection{Obfuscation Scheme design}

In this section, we design our obfuscation scheme, which consists of three core algorithms that set the theoretical foundation for \name{}: (1) \textbf{Gen} (key generation), which creates cryptographic keys containing decoy frequencies; (2) \textbf{Enc} (encryption), which physically generates obfuscation by actuating decoy signals; and (3) \textbf{Dec} (decryption), which enables authorized sensors to recover the true vital sign by filtering out known decoys. We then show the correctness of the framework and analyze the privacy guarantees of the scheme.

\subsubsection{Key Generation}

The key generation algorithm Gen produces a set of decoy signal frequencies that will be used to obfuscate the user's true vital signs. The algorithm samples $p$ frequencies from a physiologically plausible range $S$ (e.g., 60-100 BPM for heart rate), ensuring that the generated decoy frequencies are indistinguishable from actual vital signs. These frequencies are stored in a cryptographic key $k = (f_1, \ldots, f_p)$ that is shared between the user and authorized sensors.

\begin{algorithm}
\caption{Key Generation}
\label{alg:keygen}
\begin{algorithmic}[1]
\Procedure{Gen}{$p, S$}
    \State \textbf{Input:} Number of decoys $p$, Physiological range $S$ (e.g., 60-100 BPM)
    \State \textbf{Output:} Key $k = (f_1, \ldots, f_p)$
    \For{$i = 1$ to $p$}
        \State Sample frequency $f_i$ from range $S$
        \State Add $f_i$ to key $k$
    \EndFor
    \State \Return $k$
\EndProcedure
\end{algorithmic}
\end{algorithm}

By sampling decoys from the plausible heart rate range $S$, we prevent statistical attacks where adversaries might identify outliers based on physiological implausibility. All generated frequencies appear as valid vital signs, making them indistinguishable from the user's actual heartbeat without knowledge of the key $k$.

\subsubsection{Physical Obfuscation (Encryption)}

The encryption algorithm Enc physically generates the obfuscation by actuating a pneumatic device at the decoy frequencies specified in the key $k$. This process creates real physical motion co-located with the user's body that wireless sensors detect alongside the user's natural vital signs. The input to Enc is the key $k$ containing the decoy frequencies and the user's true vital sign signal $m$. The output is an obfuscated signal $c$ that represents the composite physical motion observed by any wireless sensor. Critically, this is not a digital encryption, it is a physical process where the actuator generates periodic motions at frequencies $f_1, \ldots, f_p$, which superimpose with the user's natural heartbeat motion at frequency $m$.

\begin{algorithm}
\caption{Physical Signal Obfuscation}
\label{alg:encrypt}
\begin{algorithmic}[1]
\Procedure{Enc}{$k, m$}
    \State \textbf{Input:} Key $k = (f_1, \ldots, f_p)$, True signal $m$
    \State \textbf{Output:} Obfuscated signal $c$
    \State Activate pneumatic actuator at frequencies $f_1, \ldots, f_p$
    \State Physically superimpose actuated frequencies on $m$ to create $c$
    \State \Return $c$ \Comment{$c$ contains dominant frequencies: $\{m, f_1, \ldots, f_p\}$}
\EndProcedure
\end{algorithmic}
\end{algorithm}

The resulting obfuscated signal $c$ contains $p+1$ dominant frequency components: the true vital sign $m$ plus $p$ decoy frequencies. Both authorized and unauthorized sensors observe the same physical phenomenon, a composite motion signal with multiple periodic components. The critical difference is that unauthorized sensors cannot determine which of these $p+1$ frequencies represents the true vital sign, while authorized sensors possess the key $k$ that identifies the decoy frequencies. We discuss the physical implementation details of the pneumatic actuator in Section~\ref{sec:implementation}.

\subsubsection{Signal Recovery (Decryption)}

The decryption algorithm Dec enables authorized devices to recover the user's true vital sign from the obfuscated signal $c$. Given the key $k$ that specifies the decoy frequencies, Dec applies a series of notch filters and band-stop filters centered at each decoy frequency $f_1, \ldots, f_p$. This filtering process removes the known decoy components from the composite signal, leaving only the true vital sign $m$. The algorithm takes as input the key $k$ and the obfuscated signal $c$ observed by the sensor, and outputs the recovered true signal $m$.

\begin{algorithm}
\caption{Authorized Signal Recovery}
\label{alg:decrypt}
\begin{algorithmic}[1]
\Procedure{Dec}{$k, c$}
    \State \textbf{Input:} Key $k = (f_1, \ldots, f_p)$, Obfuscated signal $c$
    \State \textbf{Output:} True signal $m$
    \State Apply band-stop filters at frequencies $f_1, \ldots, f_p$
    \State $m \leftarrow c \setminus k$ \Comment{Remove decoy frequencies}
    \State \Return $m$
\EndProcedure
\end{algorithmic}
\end{algorithm}

In practice, this is implemented through cascaded notch or band-stop filtering in the frequency domain. Each filter is designed with a narrow bandwidth centered at a decoy frequency, ensuring that it removes the decoy component while preserving the true vital sign and minimizing distortion. This filtering approach is general and works regardless of what signal processing algorithm the sensor uses for vital sign extraction, as the decoy removal happens at the fundamental signal level before any algorithm-specific processing.

\subsubsection{Correctness of the Scheme}
\label{sec:correctness}

A scheme is correct if an authorized user $\mathcal{V}$ (who possesses the key $k$) can always recover the original message $m$ from the observed obfuscated ciphertext $c$.
Formally, we must show that for any message $m \in S$ and any key $k$ generated by $\texttt{Gen}(p, S)$, the following holds:
\begin{equation}
    \texttt{Dec}(k, \texttt{Enc}(k, m)) = m
\end{equation}

\begin{proof}
The proof follows directly from the definitions of the algorithms. The encryption process $\texttt{Enc}(k, m)$ is a physical process that produces the observed multiset $c$ of frequencies. As defined in Algorithm~\ref{alg:encrypt}, this multiset is $c \coloneqq \{m\} \cup k$. The decryption process $\texttt{Dec}(k, c)$ takes $c$ and $k$ as input and, as defined in Algorithm~\ref{alg:decrypt}, computes the multiset difference $m' \coloneqq c \setminus k$. 

By substituting the definition of $c$ into the decryption operation, we obtain:
$$
m' = (\{m\} \cup k) \setminus k
$$
By the definition of multiset difference, this operation removes all $p$ elements of $k$ from the multiset, leaving only the single element $m$, thus $m' = m$.

This correctness holds even in the negligible-probability collision case where $m = p_i$ for some $p_i \in k$. In such cases, the multiset $c$ would contain two instances of the same value, and the decryption operation $c \setminus k$ correctly removes one instance (the decoy) while preserving the other (the message). Therefore, the scheme is correct.
\end{proof}

\subsection{Privacy Guarantee}
\label{sec:privacy_guarantee}

Our security objective is to ensure that an adversary cannot distinguish the user's true signal $m$ from the $p$ decoy signals. We prove that the adversary gains no meaningful advantage in identifying which of the $p+1$ observed signals represents the true message.

\subsubsection{Defining Privacy}

We formalize the adversary's task as follows: after observing the obfuscated signal $c$ containing $p+1$ frequency components, the adversary $\mathcal{A}$ must guess which one corresponds to the true vital sign. The adversary outputs an index $j \in \{1, \ldots, p+1\}$ representing their guess.

\begin{definition}[Adversary's Advantage]
A random guess succeeds with probability $1/(p+1)$. We measure the adversary's capability by their \emph{advantage}—how much better they perform compared to random guessing:
\begin{equation}
    \epsilon_{\text{adv}} \coloneqq P(\mathcal{A}\text{'s guess is correct}) - \frac{1}{p+1}
\end{equation}
\end{definition}

\noindent An obfuscation scheme is considered secure if $\epsilon_{\text{adv}}$ is bounded by a negligible value. Our goal is to prove that $\epsilon_{\text{adv}}$ is negligibly small.

\subsubsection{Analysis Framework}

To analyze the adversary's advantage, we partition all possible ciphertexts into two categories based on whether frequency collisions occur. A \emph{collision} happens when the true signal frequency $m$ coincidentally equals one of the decoy frequencies in the key $k$.

\paragraph{Non-colliding ("Good") Ciphertexts $\mathcal{C}_{\text{good}}$.} 
In this case, all $p+1$ signals (the true signal $m$ and $p$ decoys) have distinct frequencies. This is the typical operational scenario. For example, if the true heart rate is 75 bpm and we generate 3 decoys at 68, 82, and 91 bpm, the observed ciphertext is $c = \{68, 75, 82, 91\}$ with 4 distinct values. An adversary cannot determine which frequency corresponds to the true signal without the key.

\paragraph{Colliding ("Bad") Ciphertexts $\mathcal{C}_{\text{bad}}$.}
In this rare case, at least two signals have the same frequency—a collision between the true signal and a decoy.For example, if the true heart rate is 75 bpm and by chance a decoy is also generated at 75 bpm, the observed ciphertext is $c = \{68, 75, 75, 91\}$. The repeated value could potentially leak information: an adversary might reason that "75 appears twice, so it's more likely that one is real and one is a decoy," potentially gaining an advantage.

The key insight of our analysis is that collisions are extremely rare. For typical parameters, the collision probability $\delta$ is approximately $10^{-4}$. Therefore, we only need to bound the adversary's advantage in this unlikely case to prove overall security.

\subsubsection{Privacy in the Good Case}
\label{sec:proof_good_case}

We first analyze the standard operational case where all $p+1$ signals in $c$ are distinct. This occurs with overwhelming probability $1-\delta$.

\begin{lemma}
\label{lem:good_case}
For any non-colliding ciphertext $c \in \mathcal{C}_{\text{good}}$, the adversary's advantage is exactly zero.
\end{lemma}

\begin{proof}
An optimal adversary will use Bayesian inference to find the signal $s_j \in c$ that maximizes the posterior probability $P(M=s_j \mid c)$, where $M$ denotes the true message. By Bayes' rule:
\begin{align*}
    P(M=s_j \mid c) &\propto P(c \mid M=s_j) \cdot P(M=s_j) \\
    &\propto P(\text{key is } c \setminus \{s_j\}) \cdot P(\text{message is } s_j)
\end{align*}

Let $f_S(x)$ denote the probability density $P(X=x)$ for a signal $X$ sampled from the physiological distribution $\mathcal{D}_H(\cdot \mid S)$. Since both the true message and all decoys are independently sampled from the same distribution over $S$, we have:
\begin{align*}
    P(M=s_j \mid c) &\propto \left( \prod_{i \ne j} f_S(s_i) \right) \cdot f_S(s_j) \\
                   &\propto \prod_{k=1}^{p+1} f_S(s_k)
\end{align*}

The final product $\prod_{k=1}^{p+1} f_S(s_k)$ is a constant for any given ciphertext $c$, regardless of which signal $s_j$ is hypothesized as the true message. Therefore, the posterior probability $P(M=s_j \mid c)$ is identical for all $j \in \{1, \ldots, p+1\}$.

Since all signals are equally likely to be the true message, the adversary's optimal strategy is to guess uniformly at random. Their success probability is exactly $1/(p+1)$, yielding zero advantage over random guessing.
\end{proof}

\subsubsection{Bounding the Collision Probability}
\label{sec:proof_bad_case}

We now bound the probability of the pathological collision case where at least two signals in $c$ share the same frequency. This is the only scenario where information leakage is possible.

\begin{lemma}
\label{lemma:2}
    The collision probability is bounded by $\frac{p(p+1)}{2N}.$
\end{lemma}

\begin{proof}
    The $p+1$ signals (one true signal and $p$ decoys) are sampled independently from a discrete frequency space $S$ with $|S| \le N$ possible values. We denote by $\delta$ the probability that any collision occurs. Using a union bound over all $\binom{p+1}{2}$ pairs of signals:
\begin{equation}
\delta = P(\mathcal{C}_{\text{bad}}) \le \binom{p+1}{2} \cdot P(\text{two signals collide})
\end{equation}

In the worst case, signals are uniformly distributed over $S$, giving $P(\text{two signals collide}) = 1/|S|$. When $|S| = N$:
\begin{equation}
\delta \le \frac{p(p+1)}{2N}
\end{equation}
\end{proof}
For any reasonably large frequency space (e.g., $N=2^{16}$ representing heart rates at 0.12 bpm resolution over the [45, 180] bpm typical range) and practical number of decoys (e.g., $p=3$), this collision probability is negligible: $\delta \approx 1.8 \times 10^{-4}$.

\subsubsection{Main Privacy Theorem}
\label{sec:proof_main}

We now combine the analyses of both cases to establish \name{}'s overall privacy guarantee.

\begin{theorem}[\name{} Privacy Guarantee]
The \name{} obfuscation scheme provides strong privacy: the adversary's advantage $\epsilon_{\text{adv}}$ is bounded by the negligible collision probability $\delta$:
\begin{equation}
\epsilon_{\text{adv}} \le \delta \le \frac{p(p+1)}{2N}
\end{equation}
\end{theorem}

\begin{proof}
We express the adversary's total advantage using the law of total probability, partitioning over whether the ciphertext is good or bad:
\begin{equation}
    \label{eq:epsadv}
    \epsilon_{\text{adv}} = P(\text{Adv} \mid \mathcal{C}_{\text{good}})P(\mathcal{C}_{\text{good}}) + P(\text{Adv} \mid \mathcal{C}_{\text{bad}})P(\mathcal{C}_{\text{bad}})
\end{equation}

We bound each term individually. First, when the ciphertext is non-colliding ($c \in \mathcal{C}_{\text{good}}$), the adversary's advantage is zero: $P(\text{Adv} \mid \mathcal{C}_{\text{good}}) = 0$, by Lemma ~\ref{lem:good_case}

Second, in the worst-case collision scenario, a collision could theoretically reveal the message's identity perfectly. The adversary's success probability would be at most 1, giving advantage at most $1 - 1/(p+1) < 1$. We conservatively bound this term by 1: $P(\text{Adv} \mid \mathcal{C}_{\text{bad}}) \le 1$.

Third, from Lemma \ref{lemma:2}, the probability of a collision is $P(\mathcal{C}_{\text{bad}}) = \delta \le \frac{p(p+1)}{2N}$.

Substituting these bounds into Equation~\ref{eq:epsadv}:
\begin{align*}
    \epsilon_{\text{adv}} &\le (0 \cdot P(\mathcal{C}_{\text{good}})) + (1 \cdot P(\mathcal{C}_{\text{bad}})) \\
    &\le P(\mathcal{C}_{\text{bad}}) \\
    &\le \delta \le \frac{p(p+1)}{2N}
\end{align*}

Since $\delta$ is negligible for large $N$, the adversary's advantage $\epsilon_{\text{adv}}$ is also negligible, proving that \name{} is a strong obfuscation scheme.
\end{proof}

%% file: Sections/Implementation.tex
\subsection{\name{} Implementation}\label{sec:implementation}

We implement \name{} using a pneumatic-based actuation system. The system consists of two main components: (1) a control unit (Arduino Uno) that generates the decoy electrical pulse signal based on the cryptographic key $k$, and (2) a pneumatic actuation device that converts the electrical pulse signal into a physical motion. 

\subsubsection{Decoy Signal Generation}
\label{sec:signal_gen}

A core requirement of our \name{} scheme is generating physical decoy signals (periodic signals of all frequencies $f_i$ of the key $k=\{f_1,..f_p\}$) that obfuscate the user's true heart-rate frequency ($m$). We generate a decoy signal using a binary pulse train that ultimately drives our pneumatic actuator (described in the next paragraph). We choose $p=3$ number of decoy frequencies, which gives the adversary's probability: $$P(\mathcal{A}\text{'s guess is correct})=\frac{1}{3+1}=0.25.$$  The pulse signal is generated through a two-step process. First, we generate a 10-second base signal at 2000~Hz sampling rate, combining multiple sinusoids at the obfuscation frequencies (e.g., 53, 79, and 101 bpm) to create a complex composite signal, and copy 3 times to create a 30-second signal. Second, the base signal is converted to a binary pulse train using zero-crossing detection, where each positive-going zero-crossing triggers a fixed-duration pulse (25~ms). The pulse signal is used as an input to the air valve, which drives the pneumatic chamber \name{}. To visually understand the pulse signal in the frequency domain, we did a time-frequency analysis (spectrogram) with frequency resolution of 6 bpm. The spectrogram (Fig.~\ref{fig:spect_input}) shows that the generated pulse input indeed contains strong frequency components around 53, 79, and 101 bpms (labeled with dashed red lines).

\begin{figure}[htbp]
    \centering 
    
    \begin{subfigure}[b]{0.48\linewidth}
        \centering
        \includegraphics[width=\linewidth]{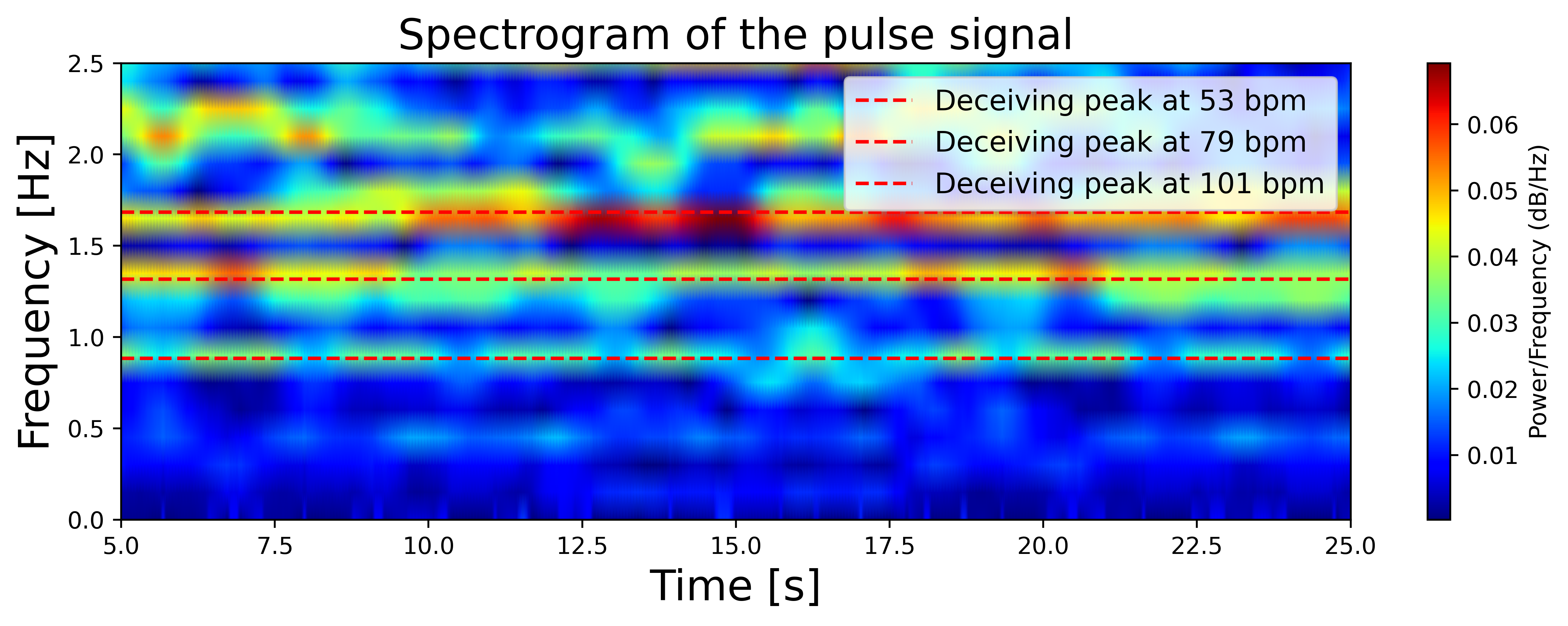}
        \caption{~}
        \label{fig:spect_input}
    \end{subfigure}
    \hfill 
    \begin{subfigure}[b]{0.48\linewidth}
        \centering
        \includegraphics[width=\linewidth]{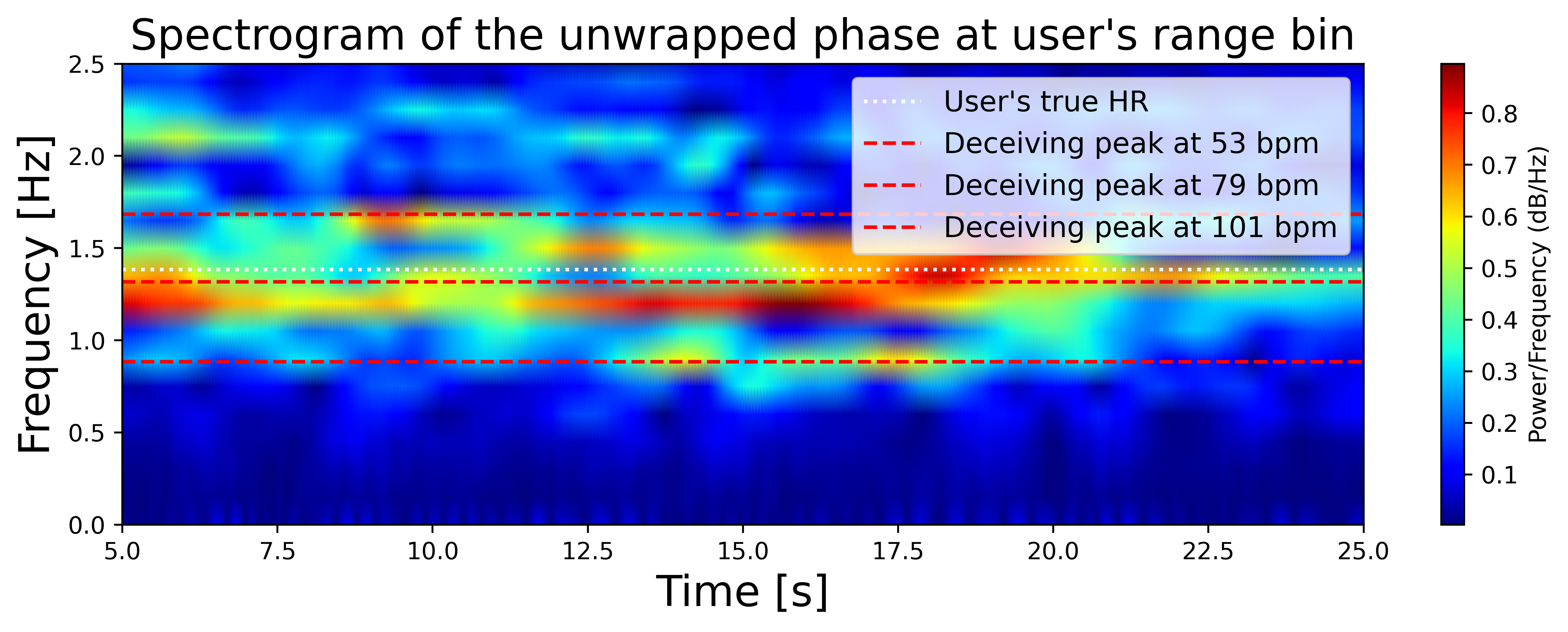}
        \caption{~ }
        \label{fig:spect_output}
    \end{subfigure}
    
    \caption{Time-frequency analysis (spectrogram) of the theoretical decoy signal and the signal collected with mmWave for a person wearing \name{}. (a) Shows high frequency content around the fake frequencies from the key $k$. (b) mmWave collected data (User 4, trial 8) shows high frequency content around the same frequencies, with some minor distortions (which are expected because of the pulse conversion step). }
    \label{fig:main_complex_layout}
\end{figure}


\subsubsection{Pneumatic Device}

The pneumatic system employs a 12 V diaphragm pump with 6 L/min flow rate as the main air input. Using an Arduino and an nMOS transistor, the pulse signal from the previous step drives a normally-open air valve, enabling rapid inflation period (25 ms) and natural deflation period, creating the periodic expansion-contraction motion that mimics physiological patterns. The Arduino is interfaced using a Macbook Pro computer and the pump is powered-up using a DC power supply. The complete hardware is shown in Figure \ref{fig:ckt}.

\begin{figure*}[htbp]
    \centering 
    \includegraphics[width=\textwidth]{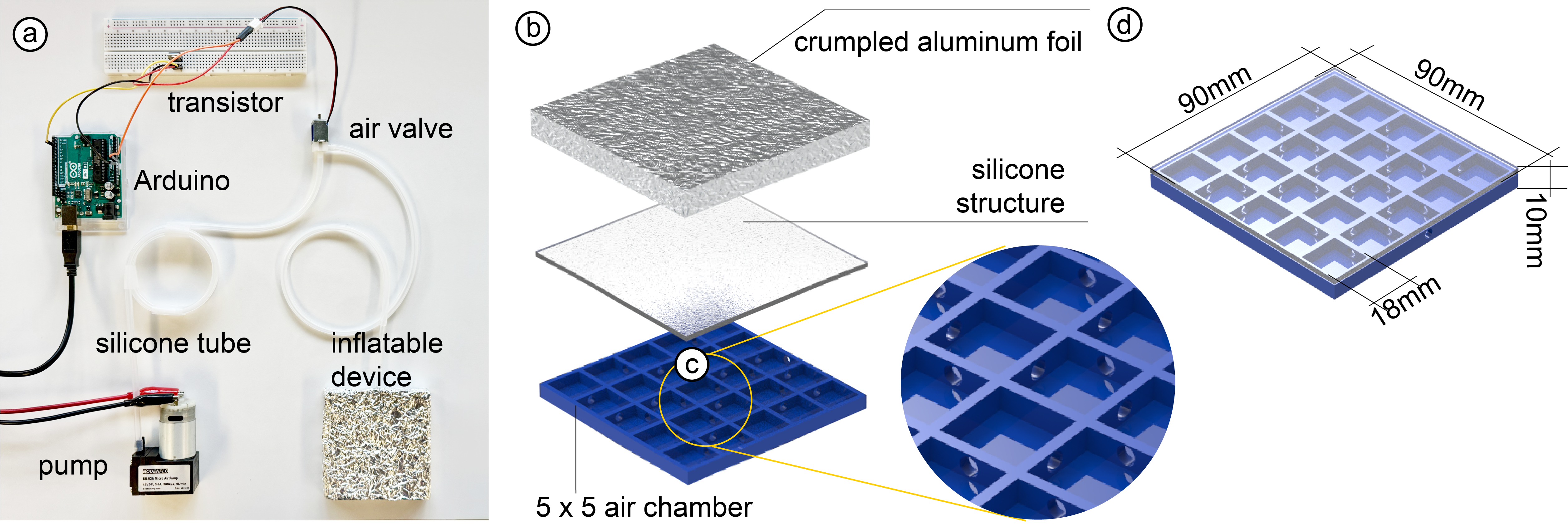}
        \caption{(a) Complete hardware for \name{} (not powered-up). (b) Inflatable pneumatic device building blocks. (c) Interconnected airways inside the inflatable device. (d) Dimensions of the inflatable device.}
        \label{fig:ckt}
\end{figure*}

The inflatable device generates obfuscation signals through controlled pneumatic actuation. The actuator transforms from a flat configuration (approximately 10mm thick) to an inflated state (30mm displacement), compact enough to fit in a chest pocket while providing sufficient radar cross-section for detection. The actuator consists of a multi-chamber silicone structure (Ecoflex 0050) with interconnected airways that enable uniform inflation (Fig. \ref{fig:ckt}c). The outer surface is coated with aluminum foil to enhance radar reflectivity, ensuring strong signal returns for both mmWave and acoustic sensing modalities (Fig. \ref{fig:ckt}b). The chamber design features a $5\times5$ grid of cells (each 18 mm$\times$18 mm) connected by 5 mm airways, allowing rapid pressure equalization while maintaining structural integrity during repeated inflation cycles (Fig. \ref{fig:ckt}d). When activated, the coordinated pump and valve operation creates periodic expansion-contraction cycles at the programmed obfuscation frequencies. The inflation phase (pump on, valve closed) lasts 25 ms while deflation (pump off, valve open) occurs within 50ms, enabling operation across the full heart rate frequency range (0.8-3.0 Hz). This pneumatic approach generates physical motion detectable by all wireless sensing modalities while maintaining a simple, reliable design suitable for extended operation.

We validate the pneumatic system using our data collection mmWave device (radar configuration detailed in the next section).  To visually understand the collected signal in the frequency domain, we again did a time-frequency analysis (spectrogram) with frequency resolution of 6 bpm. Data was collected from a user wearing \name{} at a 30 cm distance from the mmWave radar. The spectrogram (Fig.~\ref{fig:spect_output}) shows that the mmWave signal sustained the strong frequency components around 53, 79, and 101 bpms from earlier (also labeled with dashed red lines), thereby validating our implementation of \name{}.

%% file: Sections/Experiments.tex
\section{Experimental Validation}

\begin{figure}[htbp]
    \centering 
    \includegraphics[]{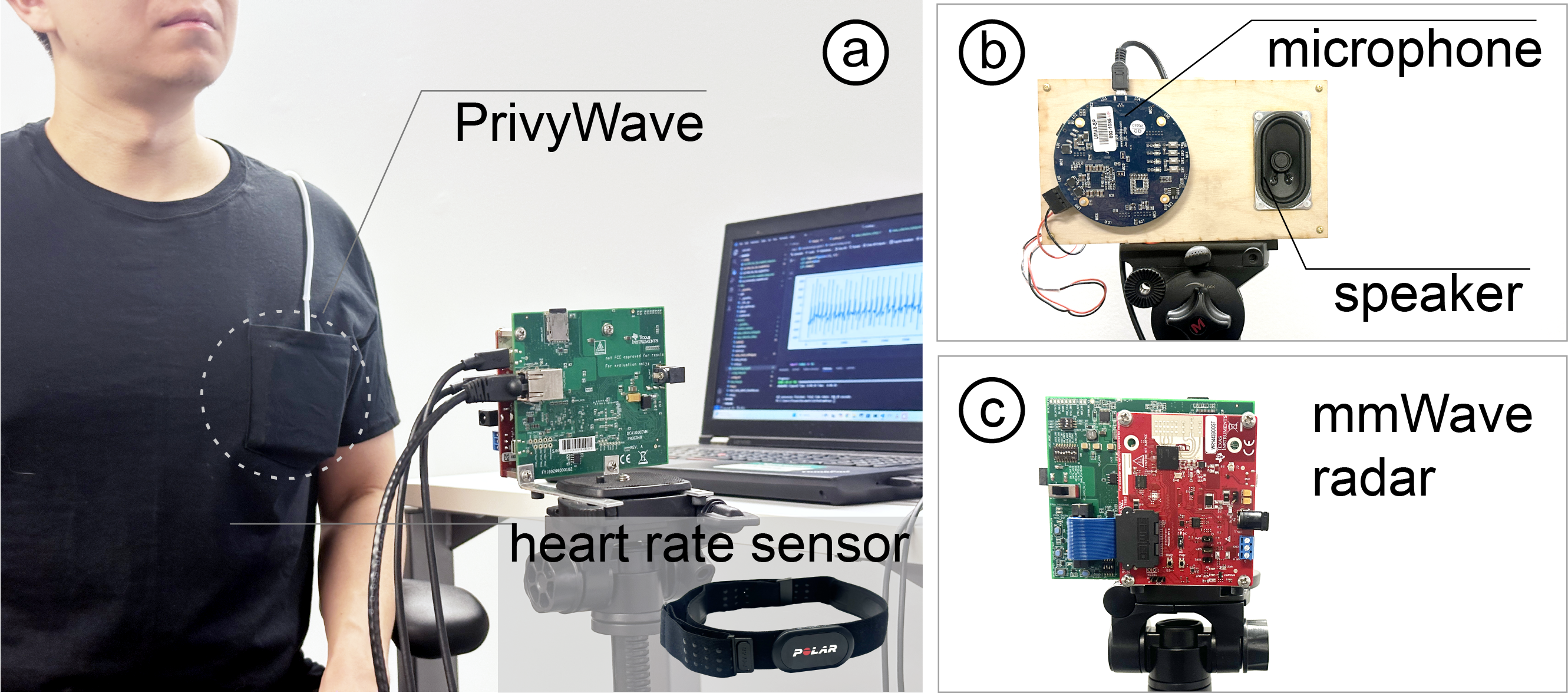}
        \caption{(a) mmWave data collection setup. (b) Acoustic sensor board with microphone array and a speaker. (c) mmWave radar sensor with data collection board.}
        \label{fig:experiment}
\end{figure}

We conduct comprehensive experiments to validate \name{}'s effectiveness across multiple scenarios: a user study with two sensing modalities (mmWave and acoustic), performance benchmarks across different environments, distances, and orientations. 

\subsection{Experiment Setup}

\begin{table}[htbp]
\centering
\caption{Comparison of Sensor System Parameters} 
\label{tab:sensor_comparison} 

\begin{subtable}{0.48\textwidth}
    \centering
    \caption{mmWave Sensor Parameters}
    \label{tab:mmwave_params} 
    \begin{tabular}{@{}lll@{}}
    \toprule
    \textbf{Parameter} & \textbf{Value} & \textbf{Unit} \\
    \midrule
    \multicolumn{3}{c}{\textit{Sensor Configuration}} \\
    TX Antennas ($N_{TX}$) & 2 &  \\
    RX Antennas ($N_{RX}$) & 4 &  \\
    Start Frequency ($f_c$) & 77 & GHz \\
    Frequency Slope ($S$) & 60.012 & MHz/$\mu$s \\
    ADC Sample Rate ($f_s$) & 5 & Msps \\
    ADC Samples ($N_{ADC}$) & 256 & samples \\
    Frame Periodicity ($T_F$) & 0.5 & ms \\
    Range FFT Size ($N_{FFT}$) & 256 & points \\
    Chirp Duration ($T_C$) & 98 & $\mu s$ \\
    \midrule
    \multicolumn{3}{c}{\textit{Calculated Performance}} \\
    Bandwidth ($B$) & 3.07 & GHz \\
    Range Resolution ($\Delta R$) & 4.88 & cm \\
    \bottomrule
    \end{tabular}
\end{subtable}
\hfill 
\begin{subtable}{0.48\textwidth}
    \centering
    \caption{Acoustic Sonar Parameters}
    \label{tab:sonar_params} 
    \begin{tabular}{@{}lll@{}}
    \toprule
    \textbf{Parameter} & \textbf{Value} & \textbf{Unit} \\
    \midrule
    \multicolumn{3}{c}{\textit{Sensor Configuration}} \\
    Transmitter (TX) & 1 & (Speaker) \\
    Receivers (RX) & 1 & (Mic.) \\
    Start Frequency ($f_{\text{start}}$) & 18 & kHz \\
    End Frequency ($f_{\text{end}}$) & 22 & kHz \\
    Sample Rate ($f_s$) & 48 & kHz \\
    ADC Samples ($N_{ADC}$) & 512 & samples \\
    Range FFT Size ($N_{FFT}$) & 512 & points \\
    Chirp Duration ($T_C$) & 10.67 & ms \\
    \midrule
    \multicolumn{3}{c}{\textit{Calculated Performance}} \\
    Bandwidth ($B$) & 4 & kHz \\
    Range Resolution ($\Delta R$) & 4.29 & cm \\
    \bottomrule
    \end{tabular}
\end{subtable}

\end{table}

Our experimental setup employs two wireless sensing systems to evaluate \name{}'s effectiveness. For mmWave radar sensing, we use the Texas Instruments IWR1443BOOST evaluation board operating at 77 GHz start frequency with 3.07 GHz bandwidth, achieving 4.88 cm range resolution (Fig. \ref{fig:experiment}c). The radar is equipped with 2 transmit and 4 receive antennas with 256-point range FFT processing. For acoustic sensing, we use a UMA-8-SP USB mic array with a speaker transmitter (Fig. \ref{fig:experiment}b), operating with 18-22 kHz FMCW chirps (4 kHz bandwidth) that achieve 4.29 cm range resolution through 512-point FFT processing. Both sensors were interfaced using a Windows 10 laptop and Python scripts. The detailed configuration of both sensors is outlined in Table~\ref{tab:sensor_comparison}. For ground truth heart rate measurement, we used a Polar H10 chest strap worn under clothing, with data captured at 130 Hz sampling rate and synchronized with the wireless sensors through python scripts.

For unauthorized detection, we implement state-of-the-art heart rate measurement algorithms for mmWave~\cite{alizadeh2019remote} and acoustic sensing~\cite{qian2018acousticcardiogram}. For authorized detection, we apply the same algorithms but first remove the known decoy frequencies using narrow band stop filters centered at $f_1, \ldots, f_p$ as specified by the key $k$. Heart rate is computed from the displacement signal using the average RR interval over the 30-second recording period:
$\text{Heart Rate (bpm)} = \frac{60}{\text{Average RR Interval}},$  and we evaluate performance using mean absolute error (MAE) between the measured heart rate and the Polar H10 ground truth.

\subsection{User Study}

\subsubsection{Participant recruitment and demographics}
We recruited 14 participants aged 22-35 years from the university campus. All participants provided informed consent, and the study was approved by our institutional review board (Protocol \#IRB0148510). One participant's data was excluded due to ground truth sensor disconnection, resulting in 13 complete datasets for analysis.

\subsubsection{Experimental Protocol}

We conducted the user study in an open laboratory environment. After a brief explanation of the study and demographic data collection, each participant was asked to wear the Polar H10 strap under their clothes for ground truth heart rate measurement and place the \name{} device in their chest pocket. For all experiments, we positioned a sensor (mmWave or acoustic) approximately 30 cm in front of the participant. We initialized \name{} with $p=3$ decoy frequencies. Two distinct key sets were generated and used across all experiments. For each sensing modality, we collected three 30-second recordings with each of the two key sets, totaling six recordings per modality. Participants rested for 10-15 seconds between recordings, resulting in 12 total recordings per participant across both modalities. 

The same recorded data was then processed under two assumptions: the \textit{authorized} case where the sensor possesses the key $k$ and filters out known decoy frequencies before heart rate estimation, and the \textit{unauthorized} case where the sensor processes the signal without the key. This controlled comparison, where both measurements are derived from identical sensor observations with access to the cryptographic key being the only variable, directly demonstrates \name{}'s selective protection capability.

\subsubsection{Results}

\begin{figure*}
    \centering
    \includegraphics[width=\textwidth]{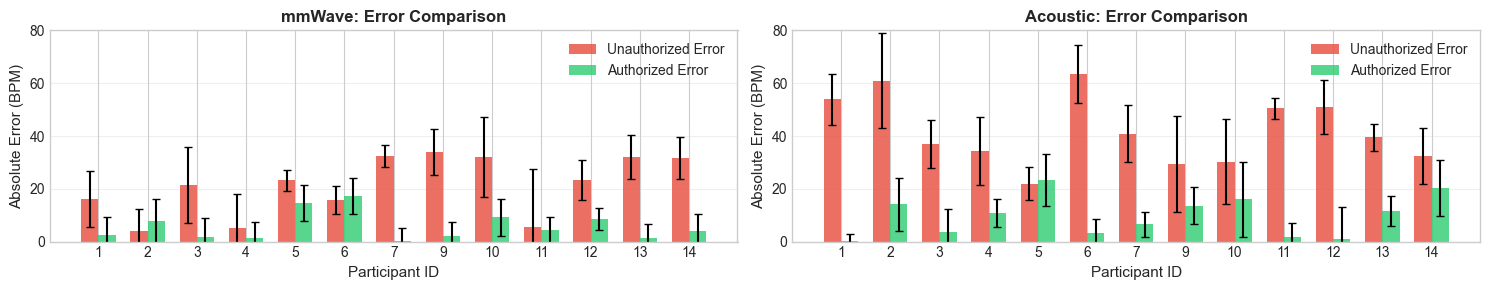}
    \caption{Individual participant heart rate errors for mmWave (left) and Acoustic (right) sensing. Red bars show the high absolute error for unauthorized devices, which consistently selected decoy frequencies. Green bars show the significantly lower error for authorized devices after decryption. Error bars represent the standard deviation across recordings.}
    \label{fig:participant_errors}
\end{figure*}

\begin{table}
\centering
\caption{Heart Rate Detection Accuracy Summary}
\label{tab:accuracy_summary}
\begin{tabular}{lcccc}
\toprule
\multirow{2}{*}{Metric} & \multicolumn{2}{c}{mmWave} & \multicolumn{2}{c}{Acoustic} \\
\cmidrule(lr){2-3} \cmidrule(lr){4-5}
& Unauthorized & Authorized & Unauthorized & Authorized \\
\midrule
MAE (BPM) & 21.3 ± 10.7 & 5.8 ± 5.2 & 42.0 ± 12.4 & 9.7 ± 7.3 \\
Median Error & 23.2 & 4.1 & 40.8 & 11.6 \\
Error Range & [4.0, 33.8] & [0.2, 17.2] & [22.0, 63.4] & [0.3, 23.3] \\
\midrule
Protection Ratio & \multicolumn{2}{c}{3.67× ($p<.001$)} & \multicolumn{2}{c}{4.33× ($p<.001$)} \\
\bottomrule
\end{tabular}
\end{table}

Figure~\ref{fig:participant_errors} illustrates the per-participant heart rate errors for both mmWave (left) and acoustic (right) sensing. The mean and median bpm errors are reported for both unauthorized and authorized devices of all users.  The unauthorized errors (red bars) are consistently high, indicating successful obfuscation, while the authorized errors (green bars) remain low, demonstrating successful signal recovery for both sensing modality. A paired-samples t-test (N=13) was conducted with these error values, and the results showed a that the authorized MAE is significantly lower than the unauthorized MAE ($p<.001$) . 
Table~\ref{tab:accuracy_summary} provides the summarized statistics for these findings. We analyze the performance for each modality:

\noindent
\textbf{For mmWave Sensing:} The unauthorized devices were effectively deceived. These sensors exhibited a high MAE of 21.3 $\pm$ 10.7 BPM, which confirms they consistently influenced by the decoy frequencies instead of the true heart rate. In contrast, the authorized device, which used the cryptographic key to filter out the decoy signals, achieved a low MAE of 5.8 $\pm$ 5.2 BPM. To put these errors into context, the reported heart-rate error of \cite{alizadeh2019remote} is $20\%$, while our errors converted into a percentage is $7.5\%$ for the authorized and $26\%$ for the unauthorized cases. So for mmWave, we are able to successfully preserve the utility of wireless heart-rate monitoring while providing obfuscation guaratee.

\noindent
\textbf{For Acoustic Sensing:} The protection was even more robust, with unauthorized devices showing a very high MAE of 42.0 $\pm$ 12.4 BPM. The authorized device's performance was slightly degraded compared to mmWave, with an MAE of 9.7 $\pm$ 7.3 BPM. The reported maximum heart-rate error of \cite{qian2018acousticcardiogram} is 3 BPM. This reduced accuracy for acoustic sensing (affecting both authorized and unauthorized measurements, as seen in the high error bars in Figure~\ref{fig:participant_errors}) is likely due to the physical properties of the modality. In our study, the users were just asked to sit in front of the system, but they did natural small movements while data collection, which also contributes to error while the reported method was more restricted during data collection. Acoustic signals have lower penetration through clothing and are more susceptible to environmental noise and multipath interference, resulting in a lower signal-to-noise ratio (SNR) compared to mmWave at the same range.

\noindent
\textbf{Protection Effectiveness:} Despite the different baseline accuracies, \name{}'s effectiveness is confirmed across both modalities. The protection ratio, defined as the unauthorized MAE divided by the authorized MAE, was 3.67$\times$ for mmWave and 4.33$\times$ for acoustic sensing. This large and statistically significant difference (mmWave: $p = 0.0011$; acoustic: $p < 0.001$) validates that our physical-layer obfuscation approach is both highly effective and modality-agnostic.

\subsection{Performance Benchmarks}

To characterize \name{}'s robustness under varying physical conditions, we conducted systematic performance benchmarks using  mmWave radar. We chose mmWave over acoustic sensing for these benchmarks because it demonstrated higher detection accuracy in our user study (5.8 BPM vs 9.7 BPM authorized error). Higher accuracy sensors present greater privacy risks to users, as they can more reliably extract heart rate. Therefore, mmWave represents the most challenging and critical case for validating \name{}'s protection effectiveness.

\subsubsection{Effect of Environment}

\begin{table}[t]
\centering
\caption{Performance Benchmark Across Different Environments (mmWave Radar)}
\label{tab:environment_test}
\begin{tabular}{@{}lccc@{}}
\toprule
\textbf{Environment} & \textbf{True HR} & \textbf{Unauth} & \textbf{Auth} \\
 & \textbf{(BPM)} & \textbf{MAE} & \textbf{MAE} \\
\midrule
Lab (Open Space) & 66.0 ± 3.0 & 6.0 ± 12.5 & 0.3 ± 1.5 \\
Kitchen & 57.2 ± 1.3 & 0.8 ± 4.5 & 1.0 ± 1.3 \\
Office & 60.2 ± 1.6 & 9.4 ± 6.7 & 3.6 ± 1.9 \\
\bottomrule
\end{tabular}
\end{table}



We evaluated \name{}'s robustness across three diverse indoor environments using mmWave radar: an open laboratory space, a kitchen, and an office. We followed the same experimental procedure as the user study, with one participant collecting six 30-second recordings (three with each key set) in each environment. Table~\ref{tab:environment_test} presents the heart rate detection performance in each environment.

\name{} demonstrates effective protection in the lab and office environments. In the lab, unauthorized sensors show 6.0 BPM error while authorized sensors achieve 0.3 BPM accuracy. In the office, unauthorized error reaches 9.4 BPM compared to 3.6 BPM for authorized sensors, maintaining clear performance separation.

However, in the kitchen environment, protection is less effective: unauthorized sensors achieve 0.8 BPM error, comparable to authorized sensors' 1.0 BPM error. The cause of this reduced effectiveness in this specific environment demonstrates that certain deployment scenarios may challenge the system's obfuscation capability.

\begin{table*}[t]
\centering
\caption{Performance Benchmarks Across Distance and Viewing Angle (mmWave Radar)}
\label{tab:benchmarks}
\begin{subtable}[t]{0.48\textwidth}
    \centering
    \caption{Range Benchmark}
    \label{tab:range_test}
    \begin{tabular}{@{}lccc@{}}
    \toprule
    \textbf{Distance} & \textbf{True HR} & \textbf{Unauth} & \textbf{Auth} \\
     & \textbf{(BPM)} & \textbf{MAE} & \textbf{MAE} \\
    \midrule
    30 cm & 66.3 ± 0.6 & 11.3 ± 9.9 & 1.3 ± 1.4 \\
    60 cm & 67.3 ± 2.5 & 16.3 ± 4.2 & 0.7 ± 4.6 \\
    90 cm & 66.0 ± 3.0 & 6.0 ± 12.5 & 0.3 ± 1.5 \\
    120 cm & 67.0 ± 5.0 & 7.7 ± 8.1 & 0.7 ± 2.5 \\
    150 cm & 69.0 ± 6.2 & 4.3 ± 20.0 & 2.7 ± 2.9 \\
    \bottomrule
    \end{tabular}
\end{subtable}
\hfill
\begin{subtable}[t]{0.48\textwidth}
    \centering
    \caption{Directional Benchmark}
    \label{tab:directional_test}
    \begin{tabular}{@{}lccc@{}}
    \toprule
    \textbf{Angle} & \textbf{True HR} & \textbf{Unauth} & \textbf{Auth} \\
     & \textbf{(BPM)} & \textbf{MAE} & \textbf{MAE} \\
    \midrule
    $-60^{\circ}$ & 59.7 ± 3.1 & 11.7 ± 2.8 & 8.8 ± 2.1 \\
    $-30^{\circ}$ & 61.3 ± 0.6 & 13.3 ± 2.0 & 4.0 ± 2.3 \\
    $0^{\circ}$ & 61.7 ± 2.5 & 10.7 ± 2.3 & 0.3 ± 2.1 \\
    $+30^{\circ}$ & 60.0 ± 4.4 & 3.3 ± 5.8 & 5.0 ± 2.0 \\
    $+60^{\circ}$ & 61.0 ± 1.0 & 7.7 ± 13.6 & 5.0 ± 3.0 \\
    \bottomrule
    \end{tabular}
\end{subtable}
\end{table*}

\subsubsection{Effect of Range}

We evaluated \name{}'s robustness across varying distances using mmWave radar. We followed the same experimental procedure as the user study with one participant, with the only variable being the sensor-to-participant distance. The participant was positioned at distances ranging from 30 cm to 150 cm at 30 cm intervals, and three 30-second recordings were collected at each distance. Table~\ref{tab:range_test} shows the heart rate detection errors for both unauthorized and authorized sensors at each distance.



\name{} maintains effective protection across all tested distances (Table~\ref{tab:range_test}). The system performs best at mid-range distances (60-90 cm), where unauthorized sensors show 6-16 BPM errors while authorized sensors maintain sub-1 BPM accuracy. At 60 cm—the optimal sensing range balancing signal strength and coverage—authorized error is only 0.7 BPM compared to 16.3 BPM for unauthorized sensors. At closer ranges (30 cm), near-field effects slightly degrade performance, while at extended distances (100 cm), both device experience signal attenuation, though authorized sensors (2.7 BPM) still significantly outperform unauthorized sensors (4.3 BPM).

\subsubsection{Effect of Orientation}

We evaluated \name{}'s performance across different orientations relative to the mmWave radar's line of sight. We followed the same experimental procedure as the user study with one participant, with the only variable being the participant's orientation relative to the radar. The participant was positioned at five angles spanning the radar's 120$^{\circ}$ field of view at 30$^{\circ}$ intervals: -60$^{\circ}$, -30$^{\circ}$, 0$^{\circ}$, +30$^{\circ}$, and +60$^{\circ}$, with three 30-second recordings collected at each angle. Table~\ref{tab:directional_test} presents the heart rate detection errors at each angle.


\name{} maintains protection across the radar's $120^{\circ}$ field of view. The system performs best at center position ($0^{\circ}$), where authorized sensors achieve 0.3 BPM error while unauthorized sensors show 10.7 BPM error. At off-center angles ($\pm 30^{\circ}$ to $\pm 60^{\circ}$), signal quality degrades due to reduced radar cross-section. Notably, at $+30^{\circ}$, the unauthorized sensor achieves a low mean error of 3.3 BPM, but with high variability (std: 5.8 BPM), while the authorized sensor shows 5.0 BPM error with stable performance (std: 2.0 BPM). At other angles, authorized sensors consistently outperform unauthorized sensors (4.0-8.8 BPM vs 7.7-13.3 BPM). This demonstrates \name{}'s effectiveness across varying user orientations, with authorized sensors providing reliable measurements even when unauthorized sensors occasionally achieve low errors through chance alignment with decoy frequencies.

%% file: Sections/Discussion.tex
\section{Discussion and Future Work}

This work demonstrates that key-based obfuscation can effectively balance the utility of ubiquitous sensing with privacy protection. The key contribution is enabling selective privacy protection of wireless sensing: authorized devices can accurately monitor heart rate while unauthorized devices are effectively prevented from extracting meaningful information. As wireless sensing proliferates in everyday environments such as offices, public transit, cafes, proactive privacy mechanisms become essential. This work establishes the technical feasibility of key-based physical obfuscation and could potentially lead to future discussion on ethics, policy, and governance of pervasive sensing technologies which we will explore in future work. Below, we discuss a few limitations of the current work and outline directions for future research.

\subsection{Motion Scenarios}

Our evaluation focused on scenarios where participants remained stationary during measurements. This design choice reflects both the technical state of wireless heartbeat sensing and the most critical privacy threats in everyday life. Existing sensing algorithms achieve highest accuracy when subjects are stationary, as motion artifacts introduce significant noise that degrades detection \cite{parralejo2025challenges}. Also, these stationary conditions align with privacy-sensitive environments where individuals are most vulnerable to unauthorized monitoring: public transportation (buses, trains, airplanes), workplaces (offices, meeting rooms), public spaces (waiting rooms, restaurants, cafes, bars). In these settings, people remain relatively stationary for extended periods, often unaware of their surroundings, while potential adversaries have stable sensing conditions and ample time to collect high-quality physiological data. For instance, a malicious actor in a coffee shop could continuously monitor customers' heart rates, or an unauthorized device in a shared office could track colleagues' stress levels throughout the workday. By demonstrating effective protection in stationary settings, we address the harder and more prevalent threat case. In the future we will verify the effectiveness of these algorithms when the subject is moving. 


\subsection{Active Side-Channel Attacks}

Our threat model assumes passive adversaries who observe wireless signals without actively probing the system. However, \name{} remains vulnerable to active side-channel attacks \cite{standaert2009introduction} where an adversary could attempt to distinguish real vital signs from decoys through active stimulus-response probing. For example, an attacker might induce a sudden startle response and observe which signal components react, since the user's physiological response would change while mechanical decoys would not. Defending against such active attacks represents a fundamental challenge that requires different approaches beyond physical obfuscation, such as detecting and responding to active probing attempts. Exploring defenses against active attacks is an potential direction for future work.

\subsection{Wearability and Form Factor}

Our current prototype requires connection to an external power source due to the pneumatic actuators' power demands, limiting mobility to stationary monitoring scenarios. Future work should focus on miniaturization and power optimization to enable battery-powered operation. Beyond power considerations, seamless integration into everyday clothing would significantly improve usability. Recent advances in smart textiles provide promising directions \cite{honnet2025fibercircuits}: actuators and circuits could be embedded within fabric layers near the chest area, or shirt buttons could be redesigned as miniature pneumatic actuators. Such integration would eliminate the need for dedicated wearable devices while maintaining obfuscation effectiveness without compromising comfort or aesthetics.

\subsection{Other Vital Signs}

While this work focused on heartbeat protection, wireless sensing can detect other involuntary physiological signals with privacy implications. Breathing generates larger movements (millimeters vs. sub-millimeter) at lower frequencies (0.1-0.5 Hz), making it more detectable but potentially easier to obfuscate with larger-displacement actuators. Beyond cardiorespiratory signals, wireless sensors can detect other health related informatoin such as tremors \cite{lo2011wireless}(Parkinson's disease, 4-12 Hz). This will revealing private health conditions that could lead to discrimination. Future work could extend obfuscation techniques for breathing and other health related information leakage. 


\section{Conclusion}

We presented \name{}, a key-based physical obfuscation system for selective privacy protection in wireless heartbeat sensing. By generating controlled decoy heartbeat signals at cryptographically-determined frequencies, our system enables authorized sensors to recover accurate measurements while unauthorized sensors cannot distinguish true signals from decoys. Our evaluation across mmWave radar and acoustic sensing demonstrates effective protection (average unauthorized error: 21.3-42.0 BPM) while maintaining high authorized accuracy (5.8-9.7 BPM). The authorized accuracy is comparable to typical wireless sensing heart rate measurement systems (ranging from 3-15 BPM) \cite{qian2018acousticcardiogram, alizadeh2019remote}, meaning \name{} does not hamper the utility of wireless sensing. The system operates across multiple sensing modalities without per-modality customization and provides formal security guarantees through cryptographic key-based decoding. This work establishes physical-layer obfuscation as a viable approach for balancing privacy and utility in pervasive health monitoring, opening new directions for privacy-preserving sensing systems.